\documentclass[12pt,a4paper]{article}
\usepackage{amssymb}
\usepackage{amsfonts}
\usepackage[onehalfspacing]{setspace}

\newtheorem{theorem}{Theorem}

\newtheorem{proposition}[theorem]{Proposition}

\newenvironment{proof}[1][Proof]{\noindent\textbf{#1.} }{\ \rule{0.5em}{0.5em}}
\input{tcilatex}
\begin{document}

\title{Calibrated Forecasts: The Minimax Proof\thanks{%
Dedicated with great admiration to David Gale, in commemoration of his 100th
birthday. This version: January 2023; previous versions: August 2018;
October 2021. The author thanks J\'{e}r\^{o}me Renault for asking about the
relation between the calibration error and the number of periods, and Benjy
Weiss for providing inequality (\ref{eq:l2}).}}
\author{Sergiu Hart\thanks{%
Einstein Institute of Mathematics, Bogen Department of Economics, and
Federmann Center for the Study of Rationality, The Hebrew University of
Jerusalem. \newline
\emph{e-mail}: \texttt{hart@huji.ac.il\ } \emph{web page}: \texttt{%
http://www.ma.huji.ac.il/hart}}}
\maketitle

\begin{abstract}
We provide a formal write-up of the simple proof (1995) of the existence of
calibrated forecasts by the minimax theorem, which moreover shows that $%
N^{3} $ periods suffice to guarantee a calibration error of at most $1/N.$
\end{abstract}

%TCIMACRO{%
%\TeXButton{References Without Numbers}{\def\@biblabel#1{#1\hfill}
%\def\thebibliography#1{\section*{References}
%\addcontentsline{toc}{section}{References}
%\list
%{}{
%\labelwidth 0pt
%\leftmargin 1.8em
%\itemindent -1.8em
%\usecounter{enumi}}
%\def\newblock{\hskip .11em plus .33em minus .07em}
%\sloppy\clubpenalty4000\widowpenalty4000
%\sfcode`\.=1000\relax\def\baselinestretch{1}\large \normalsize}
%\let\endthebibliography=\endlist}}%
%BeginExpansion
\def\@biblabel#1{#1\hfill}
\def\thebibliography#1{\section*{References}
\addcontentsline{toc}{section}{References}
\list
{}{
\labelwidth 0pt
\leftmargin 1.8em
\itemindent -1.8em
\usecounter{enumi}}
\def\newblock{\hskip .11em plus .33em minus .07em}
\sloppy\clubpenalty4000\widowpenalty4000
\sfcode`\.=1000\relax\def\baselinestretch{1}\large \normalsize}
\let\endthebibliography=\endlist%
%EndExpansion

Consider a weather forecaster who announces each day a probability $p$ that
there will be rain tomorrow. The forecaster is said to be \emph{calibrated}
if, for each forecast $p$ that is used, the relative frequency of rainy days
out of those days in which the forecast was $p$ is equal to $p$ in the long
run.

The surprising result of Foster and Vohra (1998) is that calibration can be 
\emph{guaranteed}, no matter what the weather will be. There are various
proofs of this result, and there is a large literature on calibration and
its uses; see the survey of Olszewski (2015) and the more recent paper of
Foster and Hart (2021).

A simple proof of the existence of calibrated forecasts, based on the \emph{%
minimax theorem}, was provided by the author in 1995.\footnote{%
At a lecture given by Dean Foster at the Center for Rationality of the
Hebrew University of Jerusalem; see Section 4, \textquotedblleft An argument
of Sergiu Hart," in Foster and Vohra (1998).} The basic argument is as
follows (see below for details). If the forecaster knew the strategy of the
\textquotedblleft rainmaker" (which could well be a mixed, i.e.,
probabilistic, strategy), then the forecaster could clearly get calibrated
forecasts by announcing every period the corresponding known probability of
rain. Incorporating this into a finite game (by using a finite grid and a
finite horizon) yields, by von Neumann's (1928) minimax theorem for
two-person zero-sum finite games, the existence of a strategy of the
forecaster that guarantees calibration against \emph{any} strategy of the
rainmaker. This is a striking use of the minimax theorem, since the fact
that there is a calibrated \emph{reply} to any given strategy of the
rainmaker is clear, whereas the consequence that there is a \emph{single}
strategy that is calibrated against \emph{all} strategies of the rainmaker
comes as a big surprise.\footnote{%
Indeed, Foster and Vohra had a hard time getting their paper published: they
got many desk rejections saying that the result \textquotedblleft cannot be
true" (the technical report came out in 1991, and the published paper only
seven years later).}

More formally, consider a two-person finite game where player 1 has $m$
strategies, player 2 has $n$ strategies, and $u_{ij}$ is the payoff when
player 1 plays his $i$-th strategy and player 2 plays his $j$-th strategy.%
\footnote{%
It does not matter who gets this \textquotedblleft payoff" (it could be,
say, player 1's payoff); also, the game need not be a zero-sum game, as only
one payoff function is considered.} A mixed strategy $x$ of player 1 is a
probability distribution over his pure strategies $\{1,...,m\}$, i.e., $%
x=(x_{1},...,x_{m})$, where $x_{i}\geq 0$ for every $i=1,...,m$ and $%
\sum_{i=1}^{m}x_{i}=1$; similarly, a mixed strategy $y$ of player 2 is a
probability distributions over his pure strategies $\{1,...,n\}$, i.e., $%
y=(y_{1},...,y_{n}),$ where $y_{j}\geq 0$ for every $j=1,...,n$ and $%
\sum_{j=1}^{n}y_{j}=1$. When the two players play the mixed strategies $x$
and $y$, respectively, the (expected) payoff is $U(x,y):=\sum_{i=1}^{m}%
\sum_{j=1}^{n}x_{i}y_{j}u_{ij}.$ 

We now provide a formulation of the minimax theorem, in a useful yet
non-standard way.

\begin{theorem}[Minimax]
\label{th:mm}Assume\textbf{\ }that the real number $v$ satisfies the
following:

\textbf{(i)} for every mixed strategy $y$ of player 2 there is a mixed
strategy

$x\equiv x(y)$ of player 1 such that the payoff is at least $v$ (i.e., $%
U(x(y),y)\geq v$).

\noindent Then

\textbf{(ii)} there is a mixed strategy $x^{\ast }$ of player 1 that
guarantees that the payoff

is at least $v$ (i.e., $U(x^{\ast },y)\geq v$ for every mixed strategy $y$
of player 2).
\end{theorem}

Indeed, the premise (i) says that $\max_{x}U(x,y)\geq v$ for every $y$,
i.e., \linebreak $\min_{y}\max_{x}U(x,y)\geq v$; since $\max_{x}%
\min_{y}U(x,y)=\min_{y}\max_{x}U(x,y)$ by von Neumann's (1928) minimax
theorem, we get $\max_{x}\min_{y}U(x,y)\geq v$, and so, taking $x^{\ast }$
to be a maximizer there, $\min_{y}U(x^{\ast },y)\geq v$, which is the
conclusion (ii).\footnote{%
Since (ii) trivially implies (i), the two conditions (i) and (ii) are in
fact equivalent. Also, the premise (i) is easily seen to be equivalent to
\textquotedblleft for every mixed strategy $y$ of player 2 there is a \emph{%
pure} strategy $i\equiv i(y)$ of player 1 such that $U(i(y),y)\geq v$."}

Stated this way, the minimax theorem may look surprising, since from a
premise of \textquotedblleft for every $y$ there is an $x$" it gets a
conclusion of \textquotedblleft there is an $x$ such that for every $y$," a
false logical argument in general (while \textquotedblleft every child has a
mother" is true, \textquotedblleft there is a mother of all children" is
not). Nevertheless, the result is correct (and far from trivial) under the
assumptions that, first, there are finitely many pure strategies, and
second, one uses mixed strategies (the result is easily seen to be false if
either one of these assumptions fails\footnote{%
Consider the \textquotedblleft choose the higher integer" infinite game, and
the \textquotedblleft matching pennies" game with pure strategies only.}).

In the above calibration setup, the premise (i) is that for every strategy
of the rainmaker there is a strategy of the forecaster that yields a small
calibration score,\footnote{%
The \textquotedblleft calibration score" will be formally defined below, as
the average distance between forecasts and relative frequencies (and so
being calibrated means that the calibration score is equal to zero).} and
the conclusion (ii) is that there is a strategy of the forecaster that
yields a small calibration score for every strategy of the rainmaker (apply
the minimax theorem, taking as payoff the negative of the calibration
score). Let us show how to get a calibration score of, say, $10\%$. To see
that the premise (i) holds, assume that the strategy of the rainmaker is
given. We will round each forecast to a multiple of $10\%$ (the finite grid
of forecasts is thus $0\%,$ $10\%,$ ..., $100\%$); therefore, the forecast
of, say, $70\%$, is announced when the probability of rain is between $65\%$
and $75\%$. Assume that this has occurred on a large number of days so that
the law of large numbers (i.e., Chebyshev's inequality) yields an expected
error between expectation and realization of at most $5\%$; the relative
frequency of rain out of these days will then be between $65\%-5\%=60\%$ and 
$75\%+5\%=80\%$---i.e., with a calibration error of at most $10\%$. Since
the same holds for every forecast (that is used nonnegligibly often), taking
an appropriately large horizon proves the premise (i)---and thus the
conclusion (ii).

\bigskip

We now provide a formal write-up of this proof, which moreover shows that an
expected calibration error of size $\varepsilon $ is guaranteed after $%
1/\varepsilon ^{3}$ periods.

For each period (day) $t=1,2,...,$ let $a_{t}\in \{0,1\}$ be the \emph{%
weather}, with $1$ for rain and $0$ for no rain, and let $c_{t}\in \lbrack
0,1]$ be the \emph{forecast}. For convenience, we will let our forecasts lie
in the grid $D:=\{1/(2N),3/(2N),...,(2N-1)/(2N)\}$ for some positive integer 
$N$; thus, each point in $[0,1]$ is within a distance of at most $1/(2N)$
from a point in $D$ (see Remark (e) below for the standard $1/N$-grid).

The \emph{calibration score }$K_{T}$ at time $T$ is computed as follows. For
each $d\in D$ let\footnote{%
We write $\mathbf{1}_{X}$ for the indicator of the event $X;$ thus, $\mathbf{%
1}_{c_{t}=d}$ equals $1$ if $c_{t}=d$ and $0$ otherwise.} 
\[
n(d)\equiv n_{T}(d)\,%
%TCIMACRO{\TeXButton{:=}{{\;:=\;}}}%
%BeginExpansion
{\;:=\;}%
%EndExpansion
\,\sum_{t=1}^{T}\mathbf{1}_{c_{t}=d} 
\]
be the number of periods in which the forecast was $d,$ and let 
\[
\overline{a}(d)\equiv \overline{a}_{T}(d)\,%
%TCIMACRO{\TeXButton{:=}{{\;:=\;}}}%
%BeginExpansion
{\;:=\;}%
%EndExpansion
\,\frac{1}{n(d)}\sum_{t=1}^{T}\mathbf{1}_{c_{t}=d}\,a_{t} 
\]
be the (relative) frequency of rain in those $n(d)$ periods; the calibration
score $K_{T}$ is then the average distance between forecasts and rain
frequencies, namely,\footnote{%
An alternative score averages the squared errors: $\mathcal{K}%
_{T}:=\sum_{d\in D}(n(d)/T)(\overline{a}(d)-d)^{2}.$ The two scores are
essentially equivalent, because $(K_{T})^{2}\leq \mathcal{K}_{T}\leq K_{T}$
(the first inequality is by Jensen's inequality, and the second is by $%
\left\vert \overline{a}(d)-d\right\vert \leq 1$, since $\overline{a}(d)$ and 
$d$ are both in $[0,1])$.}%
\[
K_{T}\,%
%TCIMACRO{\TeXButton{:=}{{\;:=\;}}}%
%BeginExpansion
{\;:=\;}%
%EndExpansion
\,\sum_{d\in D}\left( \frac{n(d)}{T}\right) \left\vert \overline{a}%
(d)-d\right\vert . 
\]

This setup can be viewed as a finite $T$-period game in which in every
period $t=1,...,T$ the rainmaker chooses the weather $a_{t}\in \{0,1\}$ and
the forecaster chooses the forecast $c_{t}\in D,$ and the payoff is the
calibration score $K_{T}.$ Both players are assumed to have perfect recall
of past weather and forecasts (thus allowing for an \textquotedblleft
adversarial" rainmaker); since the number of periods $T$ and the sets of
choices of the players, $\{0,1\}$ and $D,$ are all finite, the game is a
finite game (i.e., each player has finitely many pure strategies).

\begin{theorem}[Calibration]
\label{th}Let $T\geq N^{3}.$ Then there exists a mixed strategy of the
forecaster that guarantees that\footnote{%
The expectation is over the random choices of the two players.} $\mathbb{E}%
\left[ K_{T}\right] \leq 1/N$ against any mixed strategy of the rainmaker.
\end{theorem}

This follows from the proposition below by applying the minimax theorem to
the payoff function $-K_{T}$.

\begin{proposition}
\label{p:minmax}Let $T\geq N^{3}.$ Then for every mixed strategy of the
rainmaker there is a strategy of the forecaster such that $\mathbb{E}\left[
K_{T}\right] \leq 1/N$.
\end{proposition}

\begin{proof}
Let $\tau $ be a mixed strategy of the rainmaker. For every $t\geq 1$ and
history $h_{t-1}=(a_{1},c_{1},...,a_{t-1},c_{t-1})\in (\{0,1\}\times
D)^{t-1} $ of rain and forecasts before time $t,$ let $p_{t}:=\mathbb{P}%
\left[ a_{t}=1|h_{t-1}\right] =\mathbb{E}\left[ a_{t}|h_{t-1}\right] $ be
the probability of rain induced by the rainmaker's strategy $\tau .$ We then
let the forecast $c_{t}$ after the history $h_{t-1}$ be the rounding of $%
p_{t}$ to the grid $D,$ with a fixed tie-breaking rule when $p_{t}$ is the
midpoint of two consecutive points in $D$; this makes $c_{t}$ a
deterministic function of the history---i.e., $c_{t}$ is $h_{t-1}$%
-measurable---and we always have $\left\vert c_{t}-p_{t}\right\vert \leq
1/(2N)$.

The calibration score $K_{T}$ can be expressed as%
\[
K_{T}=\frac{1}{T}\sum_{d\in D}\left\vert G(d)\right\vert , 
\]%
where\footnote{$G(d)$ is the difference between the actual number of rainy
days, $n(d)\overline{a}(d),$ and the predicted number of rainy days, $n(d)d,$
in the $n(d)$ days in which the forecast was $d;$ it is referred to as the
(total) \textquotedblleft gap" in Foster and Hart (2021).} 
\[
G(d)%
%TCIMACRO{\TeXButton{:=}{{\;:=\;}}}%
%BeginExpansion
{\;:=\;}%
%EndExpansion
n(d)(\overline{a}(d)-d)=\sum_{t=1}^{T}\mathbf{1}_{c_{t}=d}(a_{t}-d)=%
\sum_{t=1}^{T}\mathbf{1}_{c_{t}=d}(a_{t}-c_{t}) 
\]%
for every $d\in D$. Replacing each $c_{t}$ with $p_{t}$ yields the scores 
\begin{eqnarray*}
\widetilde{G}(d) &%
%TCIMACRO{\TeXButton{:=}{{\;:=\;}}}%
%BeginExpansion
{\;:=\;}%
%EndExpansion
&\sum_{t=1}^{T}\mathbf{1}_{c_{t}=d}(a_{t}-p_{t})\text{\ \ and} \\
\widetilde{K}_{T} &%
%TCIMACRO{\TeXButton{:=}{{\;:=\;}}}%
%BeginExpansion
{\;:=\;}%
%EndExpansion
&\frac{1}{T}\sum_{d\in D}\left\vert \widetilde{G}(d)\right\vert ;
\end{eqnarray*}%
since $|c_{t}-p_{t}|\leq 1/(2N)$ it follows that $|G(d)-\widetilde{G}%
(d)|\leq n(d)/(2N)$ and 
\begin{equation}
\left\vert K_{T}-\widetilde{K}_{T}\right\vert \leq \frac{1}{T}\sum_{d\in D}%
\frac{n(d)}{2N}=\frac{1}{2N}  \label{eq:round}
\end{equation}%
(because $\sum_{d}n(d)=T$).

We claim that\footnote{%
If one does not care about the bound $N^{3}$ on $T$ one may use at this
point various simpler Chebyshev or law-of-large-numbers inequalities (see
also Remarks (c) and (d) below).}%
\begin{equation}
\mathbb{E}\left[ \widetilde{G}(d)^{2}\right] \leq \frac{1}{4}\mathbb{E}\left[
n(d)\right]  \label{eq:l2}
\end{equation}%
for each $d\in D$. Indeed, $\widetilde{G}(d)=\sum_{t=1}^{T}\mathbf{1}%
_{c_{t}=d}Z_{t}$ where $Z_{t}:=a_{t}-p_{t}$, for which we have $\mathbb{E}%
\left[ Z_{t}|h_{t-1}\right] =0$ (because $p_{t}=\mathbb{E}\left[
a_{t}|h_{t-1}\right] )$ and $\mathbb{E}\left[ Z_{t}^{2}|h_{t-1}\right] \leq
1/4$ (because this is the variance of a Bernoulli random variable, namely, $%
a_{t}|h_{t-1}$). Then, for $s<t$ we get 
\begin{eqnarray*}
\mathbb{E}\left[ \left( \mathbf{1}_{c_{s}=d}\,Z_{s}\right) \cdot \left( 
\mathbf{1}_{c_{t}=d}\,Z_{t}\right) \right] &=&\mathbb{E}\left[ \mathbb{E}%
\left[ (\mathbf{1}_{c_{s}=d}\,Z_{s})\cdot (\mathbf{1}_{c_{t}=d}%
\,Z_{t})|h_{t-1}\right] \right] \\
&=&\mathbb{E}\left[ \mathbf{1}_{c_{s}=d}\,Z_{s}\mathbf{1}_{c_{t}=d}\,\mathbb{%
E}\left[ Z_{t}|h_{t-1}\right] \right] =0
\end{eqnarray*}%
(because the random variables $c_{s},$ $Z_{s},$ and $c_{t}$ are $h_{t-1}$%
-measurable), and for $s=t$ we get%
\begin{eqnarray*}
\mathbb{E}\left[ \left( \mathbf{1}_{c_{t}=d}\,Z_{t}\right) ^{2}\right] &=&%
\mathbb{E}\left[ \mathbb{E}\left[ \left( \mathbf{1}_{c_{t}=d}\,Z_{t}\right)
^{2}|h_{t-1}\right] \right] \\
&=&\mathbb{E}\left[ \mathbf{1}_{c_{t}=d}\,\mathbb{E}\left[ Z_{t}^{2}|h_{t-1}%
\right] \right] \leq \frac{1}{4}\mathbb{E}\left[ \mathbf{1}_{c_{t}=d}\right]
;
\end{eqnarray*}%
summing all these terms yields $\mathbb{E}\left[ \widetilde{G}(d)^{2}\right]
\leq (1/4)\sum_{t=1}^{T}\mathbb{E}\left[ \mathbf{1}_{c_{t}=d}\right] =(1/4)%
\mathbb{E}\left[ n(d)\right] $, which is (\ref{eq:l2}).

Therefore,%
\begin{eqnarray}
\mathbb{E}\left[ \widetilde{K}_{T}\right] &=&\frac{1}{T}\sum_{d\in D}\mathbb{%
E}\left[ \left\vert \widetilde{G}(d)\right\vert \right] \leq \frac{1}{T}%
\frac{1}{2}\sum_{d\in D}\left( \mathbb{E}\left[ n(d)\right] \right) ^{1/2} 
\nonumber \\
&\leq &\frac{1}{T}\frac{1}{2}\left( N\sum_{d\in D}\mathbb{E}\left[ n(d)%
\right] \right) ^{1/2}=\frac{1}{2}\left( \frac{N}{T}\right) ^{1/2},
\label{eq:N}
\end{eqnarray}%
where we have used $\mathbb{E}\left[ \left\vert \widetilde{G}(d)\right\vert %
\right] \leq \left( \mathbb{E}\left[ \widetilde{G}(d)^{2}\right] \right)
^{1/2}$ and (\ref{eq:l2}) for the first inequality, the Cauchy--Schwartz
inequality for the second one, and finally $\sum_{d}\mathbb{E}\left[ n(d)%
\right] =T$. When $T\geq N^{3}$ this gives $\mathbb{E}\left[ \widetilde{K}%
_{T}\right] \leq 1/(2N)$, and hence $\mathbb{E}\left[ K_{T}\right] \leq
1/(2N)+1/(2N)=1/N$ by (\ref{eq:round}).
\end{proof}

\bigskip

\noindent \textbf{Remarks. }\emph{(a) }Since the game between the rainmaker
and the forecaster is a game of perfect recall, by Kuhn's (1953)\ theorem
one can replace mixed strategies with their equivalent \emph{behavior}
strategies. A behavior strategy of the forecaster, which is referred to as a 
\emph{forecasting procedure}, consists of a separate randomization after
each history; i.e., it is a mapping from the set of histories to the set of
probability distributions on $D.$

\emph{(b) }$N^{3}$ is the right order of magnitude for the horizon $T$ that
guarantees a calibration error of $1/N$ when the forecaster rounds the rain
probabilities $p_{t}$ to the grid $D$, because if the rainmaker chooses $%
p_{t}$ to be uniform on $[0,1]$ then each one of the $N$ forecasts $d$ in $D$
is used about $T/N$ times, and so in order to get an error of $1/N$ one
needs $T/N$ to be of the order of $N^{2}$.

\emph{(c)} A tighter estimation in the proof of Proposition \ref{p:minmax}
uses $\mathbb{E}\left[ Z_{t}^{2}|h_{t-1}\right] =p_{t}(1-p_{t})$, which is
close to $d(1-d),$ instead of $\mathbb{E}\left[ Z_{t}^{2}|h_{t-1}\right]
\leq 1/4$ (recall that $\mathbb{E}\left[ Z_{t}^{2}|h_{t-1}\right] $ is the
variance of a Bernoulli$(p_{t})$ random variable); this yields $\mathbb{E}%
\left[ K_{T}\right] \leq 1/N$ for $T\ $starting approximately at $(2/3)N^{3}$%
. More precisely: let $f(d):=d^{\prime }(1-d^{\prime })$ where $d^{\prime
}=d+1/(2N)$ for $d<1/2,$ $d^{\prime }=d$ for $d=1/2,$ and $d^{\prime
}=d-1/(2N)$ for $d>1/2;$ then $|p_{t}-d|\leq 1/(2N)$ implies $%
p_{t}(1-p_{t})\leq f(d)$ (because $x(1-x)$ increases for $x<1/2$ and
decreases for $x>1/2$), and then the coefficient $1/4$ in inequality (\ref%
{eq:l2})\ may be replaced with $f(d)$. This yields%
\begin{eqnarray*}
\mathbb{E}\left[ \widetilde{K}_{T}\right] &\leq &\frac{1}{T}\sum_{d\in
D}\left( f(d)\mathbb{E}\left[ n(d)\right] \right) ^{1/2}\leq \frac{1}{T}%
\left( \sum_{d\in D}f(d)\right) ^{1/2}\left( \sum_{d\in D}\mathbb{E}\left[
n(d)\right] \right) ^{1/2} \\
&=&\frac{1}{T^{1/2}}\left( \sum_{d\in D}f(d)\right) ^{1/2}.
\end{eqnarray*}%
Now it is a straightforward computation to see that $\sum_{d\in
D}f(d)=N/6+1/4-1/(6N)$, and so for all $T\geq (2/3)N^{3}+N^{2}-(2/3)N$ we
have $\mathbb{E}\left[ \widetilde{K}_{T}\right] \leq 1/(2N)$, and thus $%
\mathbb{E}\left[ K_{T}\right] \leq 1/N$.

\emph{(d)} A looser but slightly simpler estimation in the proof of
Proposition \ref{p:minmax} that uses $n(d)\leq T$ for each $d$ instead of $%
\sum_{d}n(d)=T$ yields $\mathbb{E}\left[ \widetilde{K}_{T}\right] \leq
(1/T)(1/2)NT^{1/2}$, and so $\mathbb{E}\left[ K_{T}\right] \leq 1/N$ for $%
T\geq N^{4}$.

\emph{(e)} If instead of $D$ we were to use the standard $1/N$%
-grid\linebreak\ $D^{\prime }=\{0,1/N,2/N,...,1\}$ we would need to replace $%
N$ (the size of $D$) with $N+1$ (the size of $D^{\prime }$) in (\ref{eq:N}),
which would yield $\mathbb{E}\left[ K_{T}\right] \leq 1/N$ for $T\geq
(N+1)N^{2}=N^{3}+N^{2}$.

\emph{(f)} A \emph{lower} bound on the guaranteed calibration error as a
function of the number of periods $T$ has recently been obtained by Qiao and
Valiant (2021); it is of the order of $T^{-0.472}$ (improving on the trivial
lower bound of the order of $T^{-1/2}$, which is obtained when the rain is
an i.i.d. Bernoulli$(1/2)$ process; note that what we have shown here is an
upper bound of $T^{-1/3}$).

\emph{(g)} The minimax approach can be further used to obtain calibrated
forecasts that are \textquotedblleft calibeating," a concept introduced by
Foster and Hart (2022): they are guaranteed to beat the Brier score of any
other forecast by that forecast's calibration score. See Appendix A.2 of the 
\texttt{arxiv} version of Foster and Hart (2022).

\emph{(h)} The minimax proof does not construct a calibrated procedure; it
only shows its existence. There are various specific such constructions in
the literature, the simplest being the one in Section V of Foster and Hart
(2021).

\end{document}